\documentclass[twoside,12pt]{article}
\usepackage{indentfirst}
\usepackage{bm}
\usepackage{cite}
\usepackage{graphicx}
\usepackage{epsfig}
\usepackage{amsmath}
\usepackage{amsfonts}
\usepackage{amssymb}
\usepackage{amsthm}
\usepackage{latexsym}
\usepackage{amsmath}

\usepackage{booktabs}
\usepackage{caption}
\newtheorem{thm}{Theorem}

\usepackage{epsfig}
\usepackage{amsmath}
\usepackage{epstopdf}
\usepackage{pgf,fancyhdr}
\usepackage{float}
\begin{document}

\begin{center}
\LARGE\bf Multipartite entanglement based on realignment moments
\end{center}

\begin{center}
\rm Hui Zhao,$^{*1}$\ Shu-Ying Zhuang,$^{1}$ \ Naihuan Jing,$^{2}$ \ Mei-Ming Zhang$^{3}$
\end{center}

\begin{center}
\begin{footnotesize} \sl
$^1$ School of Mathematics, Statistics and Mechanics, Beijing University of Technology, Beijing 100124, China

$^2$ Department of Mathematics, North Carolina State University, Raleigh, NC 27695, USA

$^3$ School of Mathematics and Physics, Jiangsu University of Technology, Changzhou ,213001, China

\end{footnotesize}
\end{center}
\footnotetext{\\{*}Corresponding author, zhaohui@bjut.edu.cn}
\vspace*{2mm}

\begin{center}
\begin{minipage}{15.5cm}
\parindent 20pt\footnotesize
Based on the realignment moments of density matrix, we study parameterized entanglement criteria for bipartite and multipartite states. By adjusting the different parameter values, our criterion can detect not only  bound entangled states, but also non-positive partial transpose entangled states for bipartite quantum systems. Moreover we propose the definition of multipartite realignment moments and generalize the result of bipartite systems to obtain a sufficient criterion to detect entanglement for multipartite quantum states in arbitrary dimensions. And we further improve the conclusion to obtain another new entanglement criterion. The new method can detect more entangled states than previous methods as backed by detailed examples.
\end{minipage}
\end{center}

\begin{center}
\begin{minipage}{15.5cm}
\begin{minipage}[t]{2.3cm}{\bf Keywords:}\end{minipage}
\begin{minipage}[t]{13.1cm}
Quantum entanglement, Realignment moments, Bound entangled states
\end{minipage}\par\vglue8pt
\end{minipage}
\end{center}

\section{Introduction}
Quantum entanglement \cite{ref1} plays an essential role in quantum information processing such as quantum computation \cite{ref2}, quantum cryptography \cite{ref3,ref4} and quantum teleportation \cite{ref5}. Due to moments acting as practical tools in estimating some important properties of quantum systems, and the measurement of these moments is relatively easy,
many researchers have focused on how to use moments to detect entanglement effectively in recent years.

Positive partial transpose (PPT) criterion \cite{ref6} and matrix realignment criterion\cite{ref7,ref8} are two important entangled criteria. The PPT criterion is necessary and sufficient for the separable quantum states in $2\otimes2$ and $2\otimes3$ systems\cite{ref9}. In higher dimensional system, it has been found that there exist states which are PPT states but entangled, which are known as bound entangled states\cite{ref10}. The matrix realignment criterion illustrates that for any bipartite separable state, the trace norm of the realigned matrix is not greater than $1$. Realignment criteria can detect both non-positive partial transpose (NPT) entangled states and bound entangled states \cite{ref7,ref11}, but PPT criterion is able to detect only NPT entangled states.

Many approaches using moments to detect entanglement were developed. Elben et al. \cite{ref12} showed the partially transposed moments can be used to define a simple yet powerful test for bipartite entanglement. Neven et al. \cite{ref13} proposed a set of inequality conditions using partially transposed moments to detect entanglement. Yu et al. \cite{ref14} introduced an optimal entanglement detection criterion based on partial moments. Wang et al.\cite{ref15} proposed the necessary and sufficient conditions for separability of quantum states based on the moments with respect to any positive map, which can detect the bound entanglement states. Zhang et al. \cite{ref16} proposed a entanglement criterion using realignment moments. Then Aggarwal et al.\cite{ref17} proposed a criterion based on realignment moments to detect NPT and bound entangled states. In \cite{ref18}, the authors introduced a criterion using realignment moments to detect bipartite entangled states, which simplifies the related experimental measurements. In \cite{ref19} the authors presented a entanglement criterion by estimating the bound of realignment moments, and it can detect NPT and bound entangled states. Most criteria using realignment moments were focused on the detection of bipartite entangled states. However multipartite systems exhibit rich physical phenomena and application potential. Due to the extremely complex structure of entanglement in multipartite quantum states, the studies on multipartite entanglement are still far from being satisfied. We consider presenting the definition of multipartite realignment moments and using them to study the entanglement criteria of multipartite systems.

In this paper, we present parameterized entanglement criteria using realignment moments for bipartite and multipartite systems. The paper is organized as follows. In Section 2, we propose parameterized entanglement criterion based on realignment moments for bipartite systems. By adjusting the parameter, our bipartite criterion can detect more bound and NPT entangled states which are failed to detect by other criteria. In Section 3, we present the definition of multipartite realignment moments and derive useful and operational criteria to detect multipartite entangled states. By detailed examples, our results can detect more entangled states. Conclusions are given in Section 4.

\section{Bipartite entanglement criterion with parameter based on moments}
We first review some concepts about the realignment. Let $A=[a_{ij}]$ be an $m\times n$ matrix over the complex field $\mathbb{C}$, the vectorization of matrix $A$ is defined as
\begin{equation}
vec(A)=\left(a_{11},\cdots,a_{m1},a_{12},\cdots,a_{m2},\cdots,a_{1n},\cdots,a_{mn}\right)^\mathrm{T},
\end{equation}
where $\mathrm{T}$ denotes the transpose.
Let $\rho$ be a bipartite state over an $m\otimes n$ dimensional system represented as an $m\times m$ block matrix:
\begin{equation}
A(\rho)=[A_{ij}]_{m\times m},
\end{equation}
where $A_{ij}$ are $n\times n$ matrices. In fact, let $|ik\rangle$ be the computational basis elements of $\mathbb C^m\otimes\mathbb C^n$,
then the block matrix entry $A_{ij}$ is given by $(A_{ij})_{kl}=\langle ik|\rho|jl\rangle$, $1\leq k, l\leq n$.
We define the {\it realignment matrix} $\rho^{R}$ as the $m^2\times n^2$ matrix obtained by the realignment operation $R$:
$$\rho^{R}=(vec(A_{11}),\cdots,vec(A_{m1}),\cdots,vec(A_{1i}),\cdots,$$
\begin{equation}
vec(A_{mi}),\cdots,vec(A_{1m}),\cdots,vec(A_{mm}))^\mathrm{T}.
\end{equation}
The realignment criterion \cite{ref7} says that $\|\rho^{R}\|\leq1$ for any bipartite separable state $\rho$. Here $\|\cdot\|$ denotes the trace norm.

The realignment moments are the following traces: 
\begin{equation}
T_{k}=Tr\left[\left(\rho^{R\dagger}\rho^{R}\right)^k\right], k=1,2,\cdots,mn,
\end{equation}
where $\dagger$ denotes conjugate transpose. Suppose $\sigma_{1},\sigma_{2},\cdots,\sigma_{s}$ are the $s$ nonzero singular values of $\rho^{R}$, then
the realignment moments are given by $T_k=\sum_i\sigma_i^{2k}$. In particular,
\begin{align}
T_{1}&=Tr\left(\rho^{R\dagger}\rho^{R}\right)=\sum_{i=1}^{s}\sigma_i^2\label{T1},\\
T_{2}&=Tr\left[\left(\rho^{R\dagger}\rho^{R}\right)^{2}\right]=\sum_{i=1}^{s}\sigma_i^4. \label{T2}
\end{align}

Then we can formulate an entanglement criterion based on the realignment moments as follows.
\begin{thm}\label{thm1}
Suppose $\rho$ is a separable bipartite state, let $\Delta=(T_{1}^{2}-T_{1})^2-2(T_{1}^{2}-T_{2})T_1^2$ and
\begin{align}
V_1(a)=\sqrt{\frac{2}{a}\left[\left(1+\frac{a}{2}\right)T_{1}+\sqrt{\left(T_1^2-T_2\right)\frac{a^2}{2}+
\left(T_1^2-T_1\right)a+T_1^2}\right]}, \label{t14}
\end{align}
then we have that\\
(i) if $\Delta\leq0$, then $V_1(a)\leq1$ for $a>0$;\\
(ii)if $\Delta>0$,  then $V_1(a)\leq1$ for $0<a\leq\frac{T_1-T_1^2-\sqrt{\Delta}}{T_1^2-T_2}$ or $a\geq\frac{T_1-T_1^2+\sqrt{\Delta}}{T_1^2-T_2}$.
\end{thm}

\begin{proof} Suppose
$\rho^R$ has nonzero singular values $\sigma_{1},\sigma_{2},\cdots,\sigma_{s}> 0$. According to the realignment criterion, $\|\rho^R\|=\sum_{i=1}^{s}\sigma_i\leq1$.
\noindent Using (\ref{T1}) and (\ref{T2}), we have
\begin{align}
2\sum_{i<j}\sigma_{i}\sigma_{j}&=\|\rho^R\|^{2}-T_1\leq1-T_1, \label{t1} \\
T_{1}^{2}&=T_{2}+2\sum_{i<j}\sigma_{i}^{2}\sigma_{j}^2.
\end{align}
Now for any real number $a>0$,
\begin{equation}
\begin{split}
(a\sum_{i<j}\sigma_{i}\sigma_{j}-T_{1})^2
&=a^2(\sum_{i<j}\sigma_{i}\sigma_{j})^2-2aT_1(\sum_{i<j}\sigma_{i}\sigma_{j})+T_1^2\\
&\geq a^2(\sum_{i<j}\sigma_{i}^{2}\sigma_{j}^{2})-aT_{1}(1-T_1)+T_1^2\\
&=\left(T_1^2-T_2\right)\frac{a^2}{2}+\left(T_1^2-T_1\right)a+T_1^2, \label{t11}
\end{split}
\end{equation}
where the inequality is due to (\ref{t1}), $a>0$, and the following fact: $\sum_{i=1}^{n}x_i^{2}\leq\left(\sum_{i=1}^{n}x_i\right)^2$ for any nonnegative real numbers $x_1,x_2,\cdots,x_n$.

On the other hand, the left hand side of the first equation in (\ref{t11}) can be written as
\begin{equation}
\begin{split}
(a\sum_{i<j}\sigma_{i}\sigma_{j}-T_{1})^2&=\left[\frac{a}{2}(\|\rho^R\|^{2}-T_1)-T_1\right]^2\\
&=\left[\frac{a}{2}\|\rho^R\|^{2}-(1+\frac{a}{2})T_1\right]^2.\\
\end{split}
\end{equation}
Therefore, we get
\begin{equation}
\left[\frac{a}{2}\|\rho^R\|^{2}-(1+\frac{a}{2})T_1\right]^2\geq\left(T_1^2-T_2\right)\frac{a^2}{2}+
\left(T_1^2-T_1\right)a+T_1^2.\label{13shi}
\end{equation}
Let $F=\left(T_1^2-T_2\right)\frac{a^2}{2}+\left(T_1^2-T_1\right)a+T_1^2$. Next we will search for the range of $a$ which satisfies $F\geq0$. Note that
$F$ can be seen as a quadratic polynomial in the variable $a$, its discriminant  $\Delta=(T_{1}^{2}-T_{1})^2-2(T_{1}^{2}-T_{2})T_1^2$ and the highest
coefficient $T_1^2-T_2>0$, thus,\\
(i)when $\Delta\leq0$, $F\geq0$ always holds true for any $a>0$;\\
(ii)when $\Delta>0$, $F\geq0$ holds true for $0<a\leq \frac{T_1-T_1^2-\sqrt{\Delta}}{T_1^2-T_2}$ or $a\geq\frac{T_1-T_1^2+\sqrt{\Delta}}{T_1^2-T_2}$.\\
Under the condition of $F\geq 0$, we can rewrite \eqref{13shi} as follows.
$$
\frac{a}{2}\|\rho^R\|^{2}-(1+\frac{a}{2})T_1\geq\sqrt{\left(T_1^2-T_2\right)\frac{a^2}{2}+
\left(T_1^2-T_1\right)a+T_1^2},
$$
then,
\begin{equation}
\sqrt{\frac{2}{a}\left[\left(1+\frac{a}{2}\right)T_{1}+\sqrt{\left(T_1^2-T_2\right)\frac{a^2}{2}+
\left(T_1^2-T_1\right)a+T_1^2}\right]}\leq \|\rho^R\| \leq 1, \label{t12}
\end{equation}
\noindent
$ \|\rho^R\| \leq 1$ in \eqref{t12} is due to the realignment criterion. Obviously, the left hand side of (\ref{t12}) is $V_1(a)$. Thus we get $V_1(a)\leq 1$, where $a$ satisfy condition (i) or (ii), which proves Theorem \ref{thm1}.
\end{proof}
 Our criterion says that if the quantity $V_1(a)>1$ still within the range specified for $a$, then the state $\rho$ is entangled.
 We illustrate the efficiency of our criterion by detailed examples.

\textbf{Example 1} Consider the class of NPT entangled states in the $3\otimes3$ dimensional system defined as follows \cite{ref20}:
\begin{equation}
\rho_d=
\begin{pmatrix}
\frac{1-d}{2}&0&0&0&0&0&0&0&\frac{-11}{50}\\
0&0&0&0&0&0&0&0&0\\
0&0&0&0&0&0&0&0&0\\
0&0&0&0&0&0&0&0&0\\
0&0&0&0&\frac{1}{2}-d&\frac{-11}{50}&0&0&0\\
0&0&0&0&\frac{-11}{50}&d&0&0&0\\
0&0&0&0&0&0&0&0&0\\
0&0&0&0&0&0&0&0&0\\
\frac{-11}{50}&0&0&0&0&0&0&0&\frac{d}{2}
\end{pmatrix},\label{li1}
\end{equation}
where $\frac{1}{50}\left(25-\sqrt{141}\right)\leq d\leq \frac{1}{100}\left(25+\sqrt{141}\right)$. By calculating, we get $\triangle<0$. We choose $a=2$, then $V_1>1$ for the range $\frac{1}{50}\left(25-\sqrt{141}\right)\leq d\leq \frac{1}{100}\left(25+\sqrt{141}\right)$. So this family of NPT entangled states can be detected by our entanglement criterion. See Figure \ref{p1}.
\begin{figure}[h]
  \centering
  \includegraphics[width=10cm]{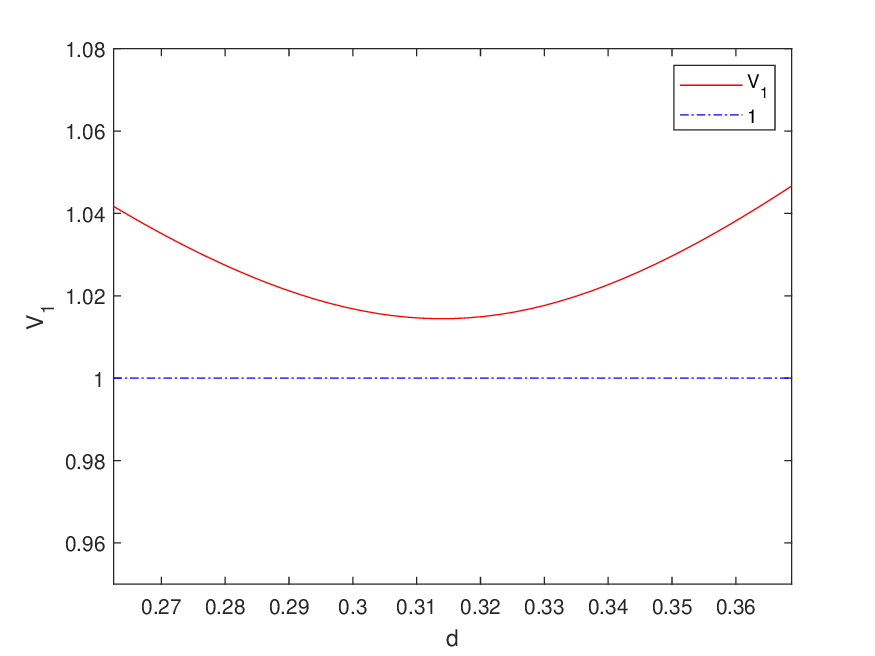}
  \caption{$V_1$ from our Theorem \ref{thm1} when $a=2$ (solid red cure), the x-axis represents the state parameter $d$.}\label{p1}
\end{figure}

\textbf{Example 2} Consider the following class of $3\otimes3$ PPT entangled states \cite{ref21}:
\begin{equation}
\rho_\varepsilon=\frac{1}{N}
\begin{pmatrix}
1&0&0&0&1&0&0&0&1\\
0&1/\varepsilon^2&0&1&0&0&0&0&0\\
0&0&\varepsilon^2&0&0&0&1&0&0\\
0&1&0&\varepsilon^2&0&0&0&0&0\\
1&0&0&0&1&0&0&0&1\\
0&0&0&0&0&1/\varepsilon^2&0&1&0\\
0&0&1&0&0&0&1/\varepsilon^2&0&0\\
0&0&0&0&0&1&0&\varepsilon^2&0\\
1&0&0&0&1&0&0&0&1
\end{pmatrix},
\end{equation}
where $\varepsilon>0$, $\varepsilon\neq1$ and $N=3(1+\varepsilon^2+\frac{1}{\varepsilon^2})$ is the normalization constant.
Direct calculation gives that $\Delta>0$, so if the state is not entangled, then
$V_1\leq 1$ on the range $0<a<0.348$ or $a>7.752$. We find that when $a=0.3$, $V_1>1$ for $\varepsilon>0$, $\varepsilon\neq1$. So Theorem \ref{thm1} can detect $\rho_\varepsilon$ is bound entangled for $\varepsilon>0$, $\varepsilon\neq1$, which is better than the result $0.622496 \leq\varepsilon\leq0.780349$ and $1.281481\leq\varepsilon\leq 1.606435$ given in \cite{ref19}.

\textbf{Example 3} Consider a family of $4\otimes4$ quantum states \cite{ref22},
\begin{equation}
\rho_{p,q}=p\sum_{i=1}^{4}|\psi_i\rangle\langle\psi_i|+q\sum_{i=5}^{6}|\psi_i\rangle\langle\psi_i| \label{li2},
\end{equation}
where $p$ and $q$ are non-negative real numbers satisfying $4p+2q=1$, and ${|\psi_i\rangle}_{i=1}^{6}$ are defined by
\begin{align*}
|\psi_1\rangle &=\frac{1}{\sqrt{2}}\left(|01\rangle+|23\rangle\right),\\
|\psi_2\rangle &=\frac{1}{\sqrt{2}}\left(|10\rangle+|32\rangle\right),\\
|\psi_3\rangle &=\frac{1}{\sqrt{2}}\left(|11\rangle+|22\rangle\right),\\
|\psi_4\rangle &=\frac{1}{\sqrt{2}}\left(|00\rangle-|33\rangle\right),\\
|\psi_5\rangle &=\frac{1}{2}\left(|03\rangle+|12\rangle\right)+\frac{|21\rangle}{\sqrt{2}},\\
|\psi_6\rangle &=\frac{1}{2}\left(-|03\rangle+|12\rangle\right)+\frac{|30\rangle}{\sqrt{2}}.
\end{align*}
When we take $p=\frac{q}{\sqrt{2}}$, the state becomes invariant under the partial transposition. Therefore, for $q_0=\frac{\sqrt{2}-1}{2}$ and $p_0=\frac{1-2q_0}{4}$, the state is PPT state. By calculation, $\Delta=0.0188$, so we focus on the range $0<a<0.203$ or $a>12.378$. Then we find out that when $a=0.2$, $V_1=1.5073$ for the state $\rho_{p_0,q_0}$, which implies $\rho_{p_0,q_0}$ is bound entangled state by our criterion. While the criterion given in \cite{ref16} did not detect this entangled state $\rho_{p_0,q_0}$. \\
When $(p,q)\neq(p_0,q_0)$, $\rho_{p,q}$ is an NPT state. Direct calculation gives that when $0\leq q<0.4609$, $\Delta>0$;  when $q\geq0.4609$, $\Delta<0$, so we need to watch for the range $0<a<0.203$ or $a>12.378$.
However we found that when $a=0.2$, $V_1>1$ in the range $0\leq q\leq\frac{1}{2}$, so Theorem \ref{thm1} implies that $\rho_{p,q}$ is NPT entangled. Thus our criterion detects that the state is an NPT entangled state when $(p,q)\neq(p_0,q_0)$.
The criterion in \cite{ref13} detects $\rho_{p,q}$ is an NPT entangled state in the range $0.425035<q\leq\frac{1}{2}$, the criterion in \cite{ref17} detects $\rho_{p,q}$ is an NPT entangled state in the range $0.00659601<q<0.153105$ and $0.26477<q\leq\frac{1}{2}$. This shows that
our criterion can detect more NPT entangled states. See Figure \ref{p2} for the comparison, the x-axis represents the state parameter $q$ and the y-axis represents the value of function with respect to $q$ in different criteria.\\
\begin{figure}[h]
  \centering
  \includegraphics[width=10cm]{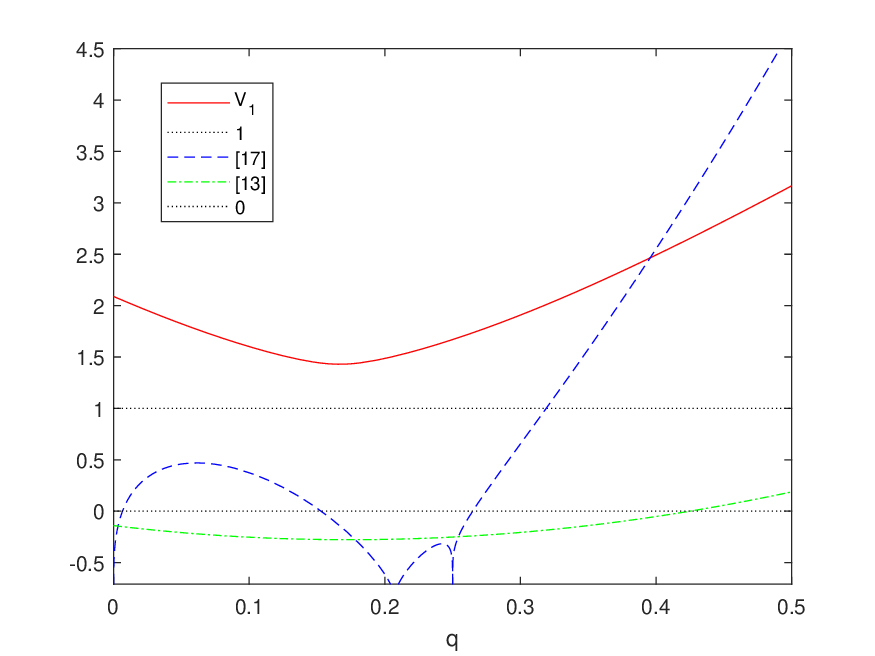}
  \caption{$V_1$ from our result when $a=0.2$ (solid red cure), criterion from Theorem $1$ in \cite{ref17} (dashed blue curve), criterion in \cite{ref13} (dash-dot green curve), the state is entangled when the value (dashed blue curve, dash-dot green curve) greater than zero.}\label{p2}
 \end{figure}

\section{Multipartite entanglement criterion based on realignment moments}
So far, a definition of the multipartite realignment moments has not been established.
In view of the applicability of the entanglement criterion in Sect. 2 for bipartite states,
we propose the notion of multipartite realignment moments and use it to
establish multipartite entanglement criteria for multipartite systems.

Firstly, we review the generalized realignment criterion in multipartite systems introduced in \cite{ref23}: if an $n$-partite density matrix $\rho$ is separable, then
\begin{equation}
\|(R_{(l)}\otimes I_{(n-l)})\rho\|\leq1, l=2,3,\cdots,n ,\label{th31}
\end{equation}
where $R_{(l)}$ refers the realignment operation on the selected $l$ subsystems while keeping the remaining $n-l$ subsystems unchanged.
Next we define the multipartite realignment moments based on (17).

Let $\rho$ be an $n$-partite state over an $d_1 \otimes d_2\otimes \cdots \otimes d_n$ dimensional system, $\rho^{R}=(R_{(l)}\otimes I_{(n-l)})\rho$.
The realignment moments of $\rho$ with respect to the operation $R_{(l)}\otimes I_{(n-l)}$ are defined by:
\begin{equation}
T_{k}^{R}=Tr\left[\left(\rho^{R\dagger}\rho^{R}\right)^k\right], k=1,2,\cdots,d_1\times d_2 \times\cdots \times d_n.
\end{equation}
Let
\begin{equation}
\Delta^R=((T_{1}^R)^{2}-T_{1}^R)^2-2((T_{1}^R)^{2}-T_{2}^R)(T_1^R)^2,
\end{equation}
\begin{equation}
V_2(u)=\sqrt{\frac{2}{u}\left[\left(1+\frac{u}{2}\right)T_{1}^R+\sqrt{\left((T_1^R)^2-T_2^R\right)\frac{u^2}{2}+
\left((T_1^R)^2-T_1^R\right)u+(T_1^R)^2}\right]}, \label{t14}
\end{equation}
where $u>0$.\\
Using a method similar to Theorem 1, we have

\begin{thm}\label{thm2}
If an arbitrary dimensional $n$-partite state $\rho$ is separable, for any operation $R_{(l)}\otimes I_{(n-l)}$ acting on $\rho$, $l=2,3,\cdots,n$, then:\\
(i)when $\Delta^R\leq0$, $V_2(u)\leq1$ holds for $u>0$;\\
(ii)when $\Delta^R>0$, $V_2(u)\leq1$ holds for $0<u\leq\frac{T_1^R-(T_1^R)^2-\sqrt{\Delta^R}}{(T_1^R)^2-T_2^R}$ or $u\geq\frac{T_1^R-(T_1^R)^2+\sqrt{\Delta^R}}{(T_1^R)^2-T_2^R}$.
\end{thm}

\textbf{Example 4} Consider the tripartite state is defined by
\begin{equation}
\rho_{GHZ,W}=q|GHZ\rangle\langle GHZ|+(1-q)|W\rangle\langle W|,
\end{equation}
where, $0\leq q\leq 1$, $|GHZ\rangle=\frac{1}{\sqrt{2}}\left[|000\rangle+|111\rangle\right]$ and $|W\rangle=\frac{1}{\sqrt{3}}\left[|001\rangle+|010\rangle+|100\rangle\right]$.
This state is entangled in the whole range of $q$  by the criterion in \cite{ref24}. Let $l=2$, then the realignment of the $\rho_{GHZ,W}$ is expressed by $\rho_{GHZ,W}^{R}=\left(R_{(2)}\otimes I_{(1)}\right)\rho_{GHZ,W}$. By calculation, we obtain that $\Delta^{R}<0$, so we choose $u=5$, then find out that
$V_2>1$ for the range of $0\leq q\leq1$. That is, our criterion is able to detect all the entanglement of $\rho_{GHZ,W}$. See Figure \ref{p3}.
\begin{figure}[h]
  \centering
  \includegraphics[width=10cm]{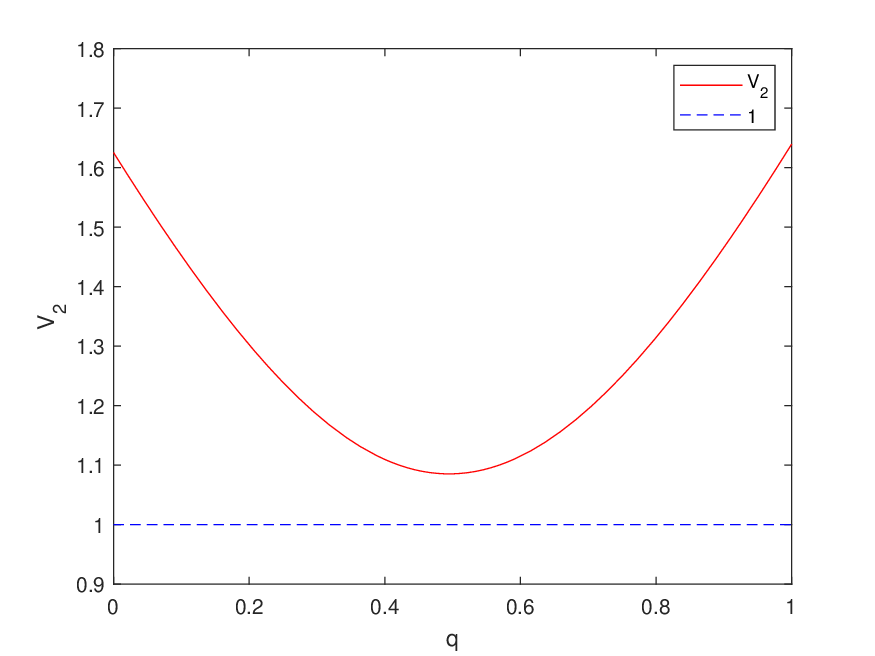}
  \caption{$V_2$ from our Theorem \ref{thm2} for the state $\rho_{GHZ,W}$ when $l=2$ and $u=5$ (solid red line), the x-axis represents the state parameter $q$.}\label{p3}
\end{figure}

We can improve the effectiveness of the criterion by introducing another quantity $V_3(v)$ to obtain the following entanglement criterion.
\begin{thm}\label{thm3}
If a multipartite state $\rho$ is separable, for any operation $R_{(l)}\otimes I_{(n-l)}$ acting on $\rho$, $l=2,3,\cdots,n$, the following inequality
\begin{equation}
V_3(v)=\sqrt{\left(\sqrt{T_1^{R}+\left(v^2+2v\right)T_2^{R}}-v\sqrt{T_2^{R}}\right)^2+\sqrt{2\left((T_1^{R})^2-T_2^{R}\right)}}\leq1,\label{th33}
\end{equation}
holds for $v>0$.
\end{thm}

\begin{proof}
Applying any operation $R_{(l)}\otimes I_{(n-l)}$ to $\rho$, $l=2,3,\cdots,n$, we get $\rho^{R}=(R_{(l)}\otimes I_{(n-l)})\rho$,
Let $\sigma_{1},\sigma_{2},\cdots,\sigma_{t}$ be the $t$ nonzero singular values of $\rho^{R}$, then
\begin{equation}
\begin{split}
T_1^R=\sum_{i=1}^{t}\sigma_i^{2}
&=\left(\sum_{i=1}^{t}\sigma_i\right)^2-2\sum_{i<j}\sigma_i\sigma_j\\
&\leq\left(\sum_{i=1}^{t}\sigma_i\right)^2-2\sqrt{\sum_{i<j}\sigma_i^2\sigma_j^2}\\
&=\left(\sum_{i=1}^{t}\sigma_i\right)^2-\sqrt{2\left(\left(\sum_{i=1}^{t}\sigma_i^2\right)^2-\sum_{i=1}^{t}\sigma_i^4\right)}\\
&=\|\rho^R\|^2-\sqrt{2\left(\left(T_1^R\right)^2-T_2^R\right)},\label{th31}
\end{split}
\end{equation}
where we have used the inequality: $\sum_{i=1}^{n}x_i^{2}\leq\left(\sum_{i=1}^{n}x_i\right)^2$ for non negative real number $x_1,x_2,\cdots,x_n$.\\
According to realignment criterion, $\|\rho^R\|=\sum_{i=1}^{t}\sigma_i \leq1$, since $\sigma_i>0$, thus $\sigma_i\leq1$, $i=1,2,\cdots,t$.\\
For any real number $v>0$,
\begin{equation}
\begin{split}
\left(T_1^R\right)^{\frac{1}{2}}+\left(v^2T_2^R\right)^{\frac{1}{2}}
&=\left(\sum_{i=1}^{s}\sigma_i^{2}\right)^{\frac{1}{2}}+\left(\sum_{i=1}^{s}(v\sigma_i^2)^2\right)^{\frac{1}{2}}\\
&\geq\left(\sum_{i=1}^{s}\left(\sigma_i+v\sigma_i^2\right)^2\right)^{\frac{1}{2}}\\
&\geq\left(\sum_{i=1}^{s}\left(\sigma_i^2+v^2\sigma_i^4+2v\sigma_i^4\right)\right)^{\frac{1}{2}}\\
&=\left(T_1^R+(v^2+2v)T_2^R\right)^{\frac{1}{2}}, \label{th32}
\end{split}
\end{equation}
where we have used the Minkowski inequality:
$\left(\sum_{i=1}^{n}x_i^{p}\right)^{\frac{1}{p}}+\left(\sum_{i=1}^{n}y_i^p\right)^{\frac{1}{p}}\geq(\sum_{i=1}^{n}(x_i+y_i)^p)^{\frac{1}{p}}$,
for $x_i,y_i>0$, and $p\geq1$. \\
Using (\ref{th31}) and (\ref{th32}), we get
\begin{equation}
||\rho^R\| \geq\sqrt{\left(\sqrt{T_1^{R}+\left(v^2+2v\right)T_2^{R}}-v\sqrt{T_2^{R}}\right)^2+\sqrt{2\left((T_1^{R})^2-T_2^{R}\right)}}.
\end{equation}
Combing $\|\rho^R\|\leq1$, we obtain (\ref{th33}). Hence the result is proved.
\end{proof}

\textbf{Example 5} Consider the four-qubit state $\rho\in H_1^{2}\otimes H_2^{2}\otimes H_3^{2}\otimes H_4^{2}$,
\begin{equation}
\rho=\frac{1-x}{16}\mathbb{I}_{16}+x|\psi\rangle\langle\psi|,~~0\leq x\leq1,
\end{equation}
where $|\psi\rangle=\frac{1}{\sqrt{2}}\left(|0000\rangle+|1111\rangle\right)$, $\mathbb{I}_{16}$ is the $16\times16$ identity matrix. Let $l=2$, then we get $\rho^{R}=(R_{(2)}\otimes I_{(2)})\rho$. By our criterion, when $v=0.01$, that $\rho$ is entangled state in the range $0.6427<x\leq1$. The result in \cite{ref25} implies that $\rho$ is entangled (not biseparable) for $\frac{17}{21}(\approx0.8095)<x\leq1$.
Theorem $4$ in \cite{ref26} says that $\rho$ is entangled under the partition $l_1l_2|l_3l_4$
for $0.915<x\leq1$ ($l_i,i=1,\cdots,4$ is the subsystem of $\rho$).
So our criterion determines a better range in detecting entangled states. See Figure \ref{p4} for the comparison, the x-axis represents the state parameter $x$ and the y-axis represents the value of function with respect to $x$ in different criteria.
\begin{figure}[h]
  \centering
  \includegraphics[width=10cm]{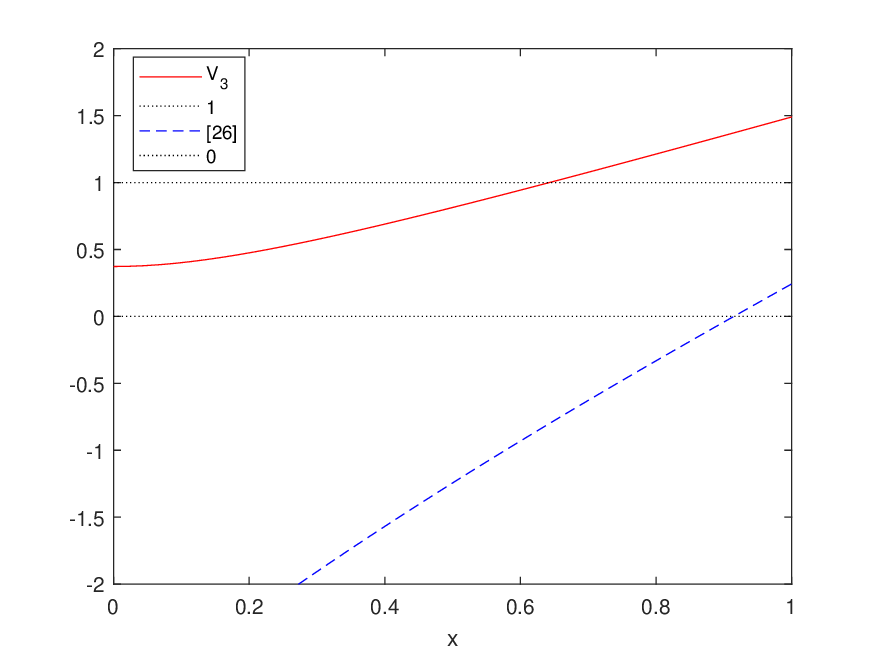}
  \caption{$V_3$ from our Theorem \ref{thm3} for the state $\rho$ when $l=2$ and $v=0.01$ (solid red line), method from Theorem $4$ in \cite{ref26} (dashed blue curve), the state is entangled when the value (dashed blue curve) greater than zero.}\label{p4}
\end{figure}

\section{Conclusion}
In this paper, we have studied the entanglement criteria based on realignment moments for arbitrary dimensional quantum systems. We have proposed parameterized
quantities to detect entanglement in bipartite systems. Detailed examples have demonstrated that our criterion can detect more bound and NPT entangled states by choosing suitable parameter than some of the existing methods.
Moreover, we have proposed a notion of realignment moments of multipartite states to study multipartite entanglement. On the basis of multipartite realignment entanglement criterion and moments, we have presented parameterized multipartite entanglement criteria, with parameters relating to realignment operation and the moments. Compared with previously available criteria, ours can detect more entangled states, and we use detailed examples
to show our findings.

\medskip
\noindent\textbf{ CRediT authorship contribution statement}

{\bf Hui Zhao, Shu-Ying Zhuang:} Formal analysis, Writing, Calculation and figure. {\bf Naihuan Jing, Mei-Ming Zhang:} Writing, review and editing.\\

\noindent\textbf{Declaration of competing interest}

The authors declare that they have no known competing financial interests or personal relationships that could have appeared to
influence the work reported in this paper.\\

\noindent\textbf{Data availability}

No data was used for the research described in the article.\\

\noindent\textbf {Acknowledgements}
This work is supported by the National Key R\&D Program of China under Grant No.
(2022YFB3806000), National Natural Science Foundation of China under Grants (12272011, 12075159,
and 12171044) and the Simons Foundation.

\end{document}